\documentclass[11pt]{article}
\usepackage{times} % REMOVE IN FINAL VERSION
\usepackage{parskip}
\usepackage{latexsym}
\usepackage{amsmath}
\usepackage{amssymb}
\usepackage{xspace}
\usepackage{algorithm}
\usepackage[noend]{algorithmic}
\usepackage{graphicx}
\usepackage{tikz}
\usetikzlibrary{shapes}
\usepackage{color}
\usepackage{epic,eepic}
\definecolor{Grey}{rgb}{0.3,0.3,0.3}
\definecolor{Black}{rgb}{0,0,0}
\definecolor{White}{rgb}{1,1,1}
\usepackage{pst-node}% http://ctan.org/pkg/pst-node
\usepackage{multido}% http://ctan.org/pkg/multido
\usepackage{multirow}
\usepackage{rotating}

\usepackage{todonotes}

\usepackage{calc,ifthen}
\newcounter{hours}
\newcounter{minutes}
\newcommand{\Printtime}{\setcounter{hours}{\time/60}%
\setcounter{minutes}{\time-\value{hours}*60}%
\thehours:%
\ifthenelse{\value{minutes}<10}{0}{}\theminutes}

\newtheorem{xtheorem}{Theorem}
\newtheorem{xcorollary}{Corollary}
\newenvironment{theorem}{\begin{xtheorem}\rm}{\end{xtheorem}}
\newenvironment{corollary}{\begin{xcorollary}\rm}{\end{xcorollary}}
\newenvironment{proof}{\begin{trivlist}\item[]{\bf Proof }}%
{\hspace*{\fill}\raisebox{-1pt}{\boldmath$\Box$}\end{trivlist}}

\newcommand{\OPT}{\ensuremath{\operatorname{\textsc{Opt}}}\xspace}
\newcommand{\ALG}{\ensuremath{\operatorname{A}}\xspace}

\newcommand{\FOURBUCKET}{\ensuremath{\operatorname{4-\textsc{Bucket}}}\xspace}

\newcommand{\FPA}{\ensuremath{\operatorname{\textsc{FPA}}}\xspace}
\newcommand{\HYBRID}{\ensuremath{\operatorname{\textsc{Hybrid}}}\xspace}

\newcommand{\GEO}{\ensuremath{\operatorname{\textsc{GreedyOpt}}}\xspace}
\newcommand{\GEOA}{\ensuremath{\operatorname{\textsc{GreedyOptAdvice}}}\xspace}
\newcommand{\GEOAC}{\ensuremath{\operatorname{\textsc{GreedyOptAdviceCancel}}}\xspace}
\newcommand{\local}{neighborhood-based\xspace}

\newcommand{\RED}{\ensuremath{\operatorname{\textsc{Red}}}\xspace}
\newcommand{\GREEN}{\ensuremath{\operatorname{\textsc{Green}}}\xspace}
\newcommand{\BLUE}{\ensuremath{\operatorname{\textsc{Blue}}}\xspace}

\newcommand{\classred}{\ensuremath{\operatorname{R}}\xspace}
\newcommand{\classgreen}{\ensuremath{\operatorname{G}}\xspace}
\newcommand{\classblue}{\ensuremath{\operatorname{B}}\xspace}

\newcommand{\colors}{\ensuremath{\operatorname{Colors}}\xspace}
\newcommand{\class}{\ensuremath{\operatorname{Class}}\xspace}
\newcommand{\borrow}{\ensuremath{\operatorname{Borrow}}\xspace}
\newcommand{\phasemin}{\ensuremath{\operatorname{Phase3Min}}\xspace}
\newcommand{\phasemax}{\ensuremath{\operatorname{Phase3Max}}\xspace}
\newcommand{\maxcolor}{\ensuremath{m}\xspace}

\newcommand{\MAX}[1]{\max\left\{#1\right\}}
\newcommand{\CEIL}[1]{\left\lceil#1\right\rceil}
\newcommand{\FLOOR}[1]{\left\lfloor#1\right\rfloor}
\newcommand{\SIZE}[1]{|#1|}
\newcommand{\SET}[1]{\ensuremath{\left\{#1\right\}}\xspace}

\newcommand{\UPPER}{\ensuremath{\operatorname{\textit{upper}}}\xspace}
\newcommand{\LOWER}{\ensuremath{\operatorname{\textit{lower}}}\xspace}

\newcommand{\REQ}[2]{\ensuremath{\operatorname{Req}(#1,#2)}}

\newcommand{\ENCODE}[1]{\ensuremath{\operatorname{\textit{enc}}(#1)}\xspace}

\newcommand{\SEQ}[1]{\ensuremath{\left\langle #1 \right\rangle}}
\newcommand{\eps}{\ensuremath{\varepsilon}\xspace}
\newcommand{\nats}{\ensuremath{\mathbb{N}}\xspace}

\usepackage{mathtools}
 \newcommand{\EXPLAIN}[2]{\underset{\mathclap{\overset{\substack{\vspace{1ex}\\\uparrow}}{\mathclap{\substack{#2}}}}}{#1}}
\newcommand{\EXPLAINUP}[2]{\overset{\mathclap{\underset{\downarrow\rule{0em}{1ex}}{\mathclap{\substack{#2}}}}}{#1}}

\newcommand{\PIND}{ \smash{\EXPLAIN{0}{\text{partition}\\\text{indicator}}} }
\newcommand{\SBIT}{ \EXPLAINUP{1}{\text{stop}\\\text{bit}} }

\begin{document}

%opening
\title{Online Multi-Coloring with Advice\,\thanks{Supported
in part by the Danish Council for Independent Research and the Villum Foundation. An extended abstract will appear in the
Twelfth Workshop on Approximation and Online Algorithms (WAOA),
Lecture Notes in Computer Science. Springer, 2014.}}

\author{Marie G. Christ \hspace{2em} Lene M. Favrholdt \hspace{2em} Kim S. Larsen \\[1ex]
        University of Southern Denmark \\
        Odense, Denmark \\[1ex]
        {\tt \{christm,lenem,kslarsen\}@imada.sdu.dk}}

\maketitle

\begin{abstract}
We consider the problem of online graph multi-coloring with advice.
Multi-coloring is often used to model frequency allocation in
 cellular networks.
We give several nearly tight upper and lower bounds for the most standard
 topologies of cellular networks, paths and hexagonal graphs.
For the path, negative results trivially carry over to bipartite
 graphs, and our positive results are also valid for bipartite graphs.
The advice given represents information that is likely to be available,
 studying for instance the data from earlier similar periods of time.
\end{abstract}

\section{Introduction}
We consider the problem of graph multi-coloring, where each node may receive
 multiple requests. 
Whenever a node is requested, a color must be assigned to the node,
 and this color must be different from any color previously assigned
 to that node or to any of its neighbors.
The goal is to use as few colors as possible.
In the online version, the requests arrive one by one, and each
 request must be colored without any information about possible future
 requests. The underlying graph is known to the online algorithm in advance.

The problem is motivated by frequency allocation in
cellular networks. These networks are formed by a number of base
transceiver stations, each of which covers what is referred
to as a cell. Due to possible interference, neighboring cells
cannot use the same frequencies.
In this paper, we use classic terminology and refer to these
cells as nodes in a graph where nodes are connected by an edge
if they correspond to neighboring cells in the network.
Frequencies can then be modeled as colors. Multiple requests for
frequencies can occur in one cell and overall bandwidth is a
critical resource.

Two basic models dominate in the discussion of cellular networks,
the highway and the city model. The former is modeled by
linear cellular networks, corresponding to paths, and the latter
by hexagonal graphs.
We consider the problem of multi-coloring such graphs.

\subsection{Analyzing online algorithms}
%In competitive analysis~\cite{ST85j,KMRS88j} of optimization problems,
%we compare online algorithms to an optimal offline algorithm.
%By definition, an \emph{online algorithm} receives the input sequence
%one item at a time, and whenever it receives an input item,
%the algorithm must make an irrevocable decision regarding
%this item without knowledge of future items and even without knowledge
%of the length of the input sequence.
%In contrast, an offline algorithm knows the entire input sequence
%before beginning its computation.
%In competitive analysis~\cite{ST85j,KMRS88j}, we measure the quality of an online algorithm
%by comparing it to an \emph{optimal offline algorithm}, \OPT.
%For problems that are inherently online, such as paging, for instance,
%\OPT is a hypothetical algorithm. Additionally,
%\OPT is assumed to have unbounded computational power.
%
%In optimization problems, one tries to minimize some cost (alternatively,
%maximize some profit). We consider a minimization problem and if
If \ALG is a multi-coloring algorithm, we let $\ALG(I)$ denote the
number of colors used by \ALG on the input
sequence $I$. %, and let $n$ denote the length of $I$.
When $I$ is clear from the context, we simply write $\ALG$ instead of $\ALG(I)$.
The quality of an online algorithm is often given in terms of the
competitive ratio~\cite{ST85j,KMRS88j}. An online multi-coloring algorithm is \emph{$c$-competitive}
if there exists a constant $\alpha$ such that for all input
sequences $I$, $\ALG(I)\leq c\OPT(I)+\alpha$.
The (asymptotic) \emph{competitive ratio} of \ALG is the infimum over all such $c$.
Results that can be established using $\alpha=0$ are referred to
as \emph{strict} (or absolute).
Often, it is a little unclear when one
refers to an \emph{optimal} online algorithm, whether this means
that the solution produced is as good as the one produced offline or
that no better online algorithm can exist.
For that reason, we may use the term \emph{strictly $1$-competitive}
to emphasize that an algorithm is as good as an optimal offline algorithm,
and \emph{optimal} to mean that no better online algorithm exists
under the given conditions.
Throughout,
% the paper
we let $n$ denote the number of requests in a given input sequence.

\subsubsection{Relaxing the concept of online}
A way of relaxing the very strict and unnatural assumption
that the algorithm has no information about the input
sequence is to give the algorithm some 
\emph{advice}. 
% The otherwise $2$-competitive online problem of
% ski rental~\cite{BE98b}
% could be optimal with just one bit of advice, 
% telling the algorithm to buy or rent skis.
The possibly most famous online problem of paging, where
no deterministic online algorithm is better than $k$-competitive on a cache size of~$k$,
can be solved optimally with just one bit of advice
per request, saying whether to keep the requested page in cache until its next request~\cite{DKP09,BKKKM09}.

A recent trend in the analysis of online algorithms has been to
consider advice, formalized under the notion of {\em advice complexity},
starting in~\cite{DKP09}.
Theoretically, results along these lines give some information
in the direction of the hardness stemming from the problem being online, relaying
information concerning how much we need to know about the future
to perform better. For practical applications, the assumption that
absolutely nothing is known about the future is often unrealistic,
and though many problems must be addressed without knowing in which
order requests arrive, quite often something is known about the
sequence of requests as a whole.

This realization that input is not arbitrary (uniformly random, for instance)
is not new, and
work focused on locality of reference in input data has tried to
capture this. Early work includes access graph results, starting
in~\cite{BIRS95}, and with references to additional
related work in~\cite{BGL12p}, but also more distributional
models, such as~\cite{AFG05}, have been developed.
An entirely different approach was initiated in~\cite{BL99j}
and further developed in~\cite{BLN01j,BFLN03j}.
The idea behind the concept of accommodating sequences is that
for many problems requiring resources, there is a close connection
between the resources available and the resources required for
an optimal offline algorithm, as when capacity of transportation systems
are matched with expected demand.
This leans itself very closely up against many of the results
that we report here, where the advice needed to do better
is often some information regarding the resources required by an optimal offline algorithm.

Thus, the results in this paper could have practical applications.
The results establish which type of information is useful,
how algorithms should be designed to exploit this information,
and what the limits are for what can be obtained.

\subsubsection{Modeling advice complexity} 
Returning to the advice complexity modeling,
some problems need very little advice.
On the other hand, complete information about the input or the desired output
is a trivial upper bound on the amount of advice needed to be optimal.
The first approach to formalizing the concept of advice measured the
number of bits per request~\cite{DKP09}.
This model is well suited for some problems where
information is tightly coupled with requests and the number of bits
needed per request is constant.
However, for most problems, we prefer the model where we simply measure
the total advice needed throughout the execution of
the algorithm.
As also discussed in~\cite{BKKKM09,HKK10}, this model avoids some modeling
 issues present in the ``per request'' 
modeling, and at the same time makes it possible to derive
sublinear advice requirements.
Thus, we use the advice model from~\cite{HKK10},
where the online algorithm has access to an infinite advice tape,
written by an offline oracle with infinite computation power.
In other words, the online algorithm can ask for the answer to any question
and read the answer from the tape.
Competitiveness is defined and measured as usual,
and the advice complexity is 
simply the number of bits read from the tape,
i.e., the maximum index of the bits read from the advice tape.

%encoding of advice
As the advice tape is infinite, we need to specify how many bits of advice the algorithm should read and if this knowledge is not implicitly available, it has to be given explicitly in the advice string. For instance, if we want \OPT
as advice (the number of colors an optimal offline algorithm uses on a given sequence, for instance),
then we cannot merely read $\CEIL{\log(\OPT+1)}$
(all $\log$s in this paper are base 2)
bits, since this would
require knowing something about the value of \OPT.
One can use a self-delimiting encoding as introduced in~\cite{E75}.
We use the variant from~\cite{BKLL12}, defined as follows: 
The value of a non-negative integer $X$ is encoded by 
a bit sequence, partitioned into three consecutive parts.
The last part is $X$ written in binary.
The middle part gives the number of bits in the last part, written in
binary.
The first part gives the number of bits in the middle part, written in
unary and terminated with a zero.
These three parts require $\CEIL{\log(\CEIL{\log (X + 1)}+1)}+1$,
$\CEIL{\log(\CEIL{\log (X + 1)}+1)}$, and $\CEIL{\log(X+1)}$ bits,
respectively,
adding a lower-order term to the number of bits of information required by an
algorithm.
%Therefore the self-delimited encoding of $X$ uses
%$\CEIL{\log(X+1)} + 2\CEIL{\log(\CEIL{\log (X + 1)}+1)}+1$ bits,
%adding a lower-order term to the number of bits of information required by an
%algorithm.
%In order to use the same notation in lower and upper bound statements,
We define
$\ENCODE{x}$ to be the minimum number of bits necessary to encode
a number $x$,
and note that the
%self-delimited
encoding above is a (good)
upper bound on $\ENCODE{x}$.
%Sometimes we do not want a number, but a sequence of bits in a situation,
%where we do not know exactly how many bits to read. Again, this must be
%encoded. This can be done similarly, and we use $\BITENCODE{b}$ to
%denote the minimum number of bits required to read $b$ bits,
%when $b$ is not known to the algorithm.

\subsection{Previous and new results}
We now discuss previous work related to multi-coloring and advice complexity and then state our results.
When working with online algorithms,
decisions are generally irrevocable, i.e., once
a color is assigned to a node, this decision is final.
%However, due to applications, cancellation of requests are sometimes
%considered, and locality is sometimes included in the modeling.
%If the algorithm is given some limited power to change the
%colors, we refer to this as {\em recoloring}.
However, in some applications, local changes of colors may be allowed
(reassignment of frequencies).
This is called {\em recoloring}.
An algorithm is {\em $d$-recoloring} if, in the
process of treating a request,
it may recolor up to a distance~$d$ away from the node of the
request.

\subsubsection{Previous results}

For multi-coloring a path, the algorithm \FOURBUCKET is
 $\frac{4}{3}$-competitive~\cite{CS10}, and this is
 optimal~\cite{CCYZZ06p}. 
Even with $0$-recoloring allowed (that is, colors at the requested node
 may be changed), \FOURBUCKET is optimal~\cite{CFL13}.
Furthermore, if
$1$-recoloring is allowed, the algorithm \GEO is strictly
 $1$-competitive~\cite{CFL13}. 

For multi-coloring bipartite graphs, the optimal asymptotic
 competitive ratio lies between $\frac{10}{7} \approx
 1.428$ and $\frac{18-\sqrt{5}}{11}\approx 1.433$~\cite{CJS13}.

In~\cite{CCYZ10}, it was shown that, for hexagonal graphs, no online
 algorithm can be better than $\frac{3}{2}$-competitive or have a
 better strict competitive ratio than $2$.
They also gave an algorithm, \HYBRID, with an asymptotic competitive
 ratio of approximately $1.9$ on hexagonal graphs.
On $k$-colorable graphs, it is strictly $\frac{k+1}{2}$-competitive,
 and hence, it has an optimal strict competitive ratio on hexagonal
 graphs.
%\HYBRID is a mixture of the $\frac{17}{7}$-competitive greedy
% algorithm~\cite{CKP08,CCYZZ07} and the $3$-competitive fixed
% allocation algorithm \FAA~\cite{N02,CCYZ10} that uses three disjoint
% color sets \RED, \GREEN, and \BLUE. Three quarters of the colors are
% used to perform \FAA and one quarter is reserved for \GREEDY. 
%
Recoloring was studied in~\cite{JKNS00j}:
No $d$-recoloring algorithm for hexagonal graphs has an asymptotic
 competitive ratio better than $1+\frac{1}{4(d+1)}$. 
For $d=0$, the lower bound was improved to~$\frac{9}{7}$.
% for hexagonal graphs with a diameter of at least~$3$.
In~\cite{SZ05j}, a $\frac{4}{3}$-competitive $2$-recoloring algorithm
is given.
The best known $1$-recoloring algorithm for hexagonal graphs is
 $\frac{33}{24}$-competitive~\cite{WZ11p}.   
%These algorithms are phase-based, and in each phase, certain subgraphs
% of the hexagonal graph are colored. 
%
For the offline problem of multi-coloring hexagonal graphs, no
 polynomial time algorithm can obtain an absolute approximation ratio
 better than $\frac43$~\cite{McDR00,NS01,NS02}, unless
 $\text{P}=\text{NP}$. 

Many other problems have been considered in the advice models,
including paging~\cite{BKKKM09}, disjoint path allocation~\cite{BBFGHKSS14},
and job shop scheduling~\cite{BKKKM09},
as well as
$k$-server~\cite{BKKK11}, knapsack~\cite{BKKR12},
set cover~\cite{KKM12}, metrical task systems~\cite{EFKR11}, and
buffer management~\cite{DHZ12}.
Also graph coloring has been considered, but in a very different
online setting, where the graph itself is not available from the
beginning. Instead, the nodes are revealed one by one and results
have been obtained for paths~\cite{FKS12}, bipartite graphs~\cite{BBHK12},
and 3-colorable graphs~\cite{SSU13}.
In~\cite{BBHKS13}, a coloring problem with restrictions going beyond the
immediate neighbors is considered.
Furthermore, there are interesting connections between advice and
 randomization and sometimes results on advice complexity can be used
 to obtain efficient random algorithms~\cite{BKKKM09,KK11j,BKKR12}.

\begin{table}[t]
\begin{center}
\begin{small}
\begin{tabular}{|c||c||c|c|c||c|c|c|} \hline
      & Ratio            & Lower & Type & Thm         
                         & Upper & Type & Thm \\ \hline\hline
\multirow{5}{*}[4ex]{\begin{sideways}\makebox[0ex]{Paths}\end{sideways}}
      & $1$              & $\log n-2$ 
                         & s & \ref{theorem:lowerbound_path}  
                         & $\log n + O(\log \log n)$ 
                         & s & \ref{theorem:bipartite-upperbound-algorithm} \\
%     & $1 + \eps$        &&
%                         && $O(\log\log n)$ 
%                         & s & Cor~\ref{theorem:oneplusepsilon}  \\
     & $1+\frac{1}{2^b}$ & $b-2$ 
                         & a & \ref{theorem:path-loweradvicebound}        
                         & $b+1 + O(\log \log n)$ 
                         & s & \ref{theorem:oneplus-competitivealgorithm}  \\
     & $<\frac{4}{3}$    & $\omega(1)$ 
                         & a & \ref{theorem:fourthirds}
                         &&&    \\ \hline\hline
\multirow{5}{*}[2ex]{\begin{sideways}\makebox[0ex]{Hexagonal}\end{sideways}}
     & $1$              &&
                         && $(n+1)\CEIL{\log n}$ 
                         & s & \ref{theorem:optimal-hex-upper} \\
     & $<\frac{5}{4}$    & $\Omega(n)$ 
                         & a & \ref{theorem:hexfivefour} 
                         &&& \\
     & $\frac{4}{3}$     &&               
                         && $n+2\SIZE{V}$ & a & \ref{theorem:43-hex}    \\
     & $\frac{3}{2}$     & $\FLOOR{\frac{n-1}{3}}$ 
                         & s & \ref{theorem:optimal-hex} 
                         & $\log n + O(\log\log n)$ 
                         & a & \ref{theorem:approx32}  \\ \hline
\end{tabular}
\end{small}
\end{center}
\caption{Overview of our results.
Recall that $n$ denotes the number of requests in the input sequence.
%The two ratios of 1 and the ratio of $\frac54$ are strict competitive
% ratios.
%The lower bound in the third result line holds for any $b$ up to $\log
% n -2$.
%Some of the competitive ratios, marked with a $\times$, are strict ratios.
We mark the ratios that are strict by ``s'' and the ones that are
asymptotic by ``a''.
%After the vertical bar, we indicate the number of the theorem
%proving the result.
Note that a strict lower bound can be larger than an asymptotic upper bound.
For each bound, we indicate the number of the theorem proving the result.
%In the second result line for hexagonal graphs, $\frac{n}{\omega(1)}$
% denotes any function which is $n$ divided by a function growing faster
% than a constant.
For readability, many of the bounds stated are weaker than those
 proven in the paper.
Moreover, the upper bounds for the path 
%are stronger than shown in that they
 hold for any bipartite graph.
The result of Theorem~\ref{theorem:fourthirds} in the third row of the
table is valid only for {\em \local} algorithms, as defined just before
Theorem~\ref{theorem:fourthirds} in Section~\ref{sec:bipartite}.
\label{table-results}}
\end{table}

\subsubsection{Our results}
An overview of our results is given in Table~\ref{table-results}.
For the path, these results are nearly tight, even with upper bounds
 that also apply to bipartite graphs.
%Note that \OPT can grow with $n$, so the
%advice is non-constant, and we show that this is necessary to get
%even below $\frac{4}{3}$.
For hexagonal graphs, note that with a linear number of advice bits,
 it is possible to be $\frac43$-competitive, and the lower bound for
 being better than $\frac54$-competitive is close to linear.
The advice given to the algorithms is essentially (an approximation
 of) \OPT or the maximum number of requests given to any clique in the
 graph.
For the underlying problem of frequency allocation, guessing these
 values based on previous data may not be unrealistic.

\section{The Path}
\label{sec:bipartite}
As explained earlier, we establish all lower bounds for paths, and since a
path is bipartite, all these negative results carry over to bipartite graphs.
Similarly, all our (constructive) upper bounds are given
for bipartite graphs and therefore also apply to paths.
We start with three lower bound results.

\subsection{Lower bounds}

\begin{theorem}\label{theorem:lowerbound_path}
 Any strictly $1$-competitive
 online algorithm for multi-coloring paths
 of at least 10 nodes has advice complexity at least
 $\CEIL{\log(\FLOOR{\frac{n}{4}}+1)}$.
\end{theorem}

\begin{proof}
 We let $m = \FLOOR{\frac{n}{4}}$ and define a set $S$ of $m+1$ sequences, all having the same prefix of length $2m$. The set $S$ will have the following property: for no two sequences in $S$ can their prefixes be colored in the same way while ending up using the optimal number of colors on the complete sequence.
Starting from one end of the path, we denote the nodes $v_1,v_2,\dots$.

We define the set $S$ to consist of the sequences $I_0, I_1, \ldots,
 I_{m}$, where $I_i$ is defined in the following way.
First $m$ requests are given to each of the nodes $v_1$ and $v_4$.
Then $i$ requests to each of $v_2$ and $v_3$.
To give all sequences the same length, the sequence ends with
% $\CEIL{ \frac{n - 2m - 2i }{2} }$ requests to $v_6$ and $\FLOOR{
% \frac{n - 2m - 2i }{2} }$ requests to $v_8$.
 $\CEIL{n-2m-2i}$ requests distributed as evenly as possible among
 $v_6$, $v_8$, and $v_{10}$.
Since $\CEIL{\CEIL{n-2m-2i}/3} \leq m$, the optimal number of colors will not be
 influenced by this part of the sequence.

Note that $\OPT(I_i)=m+i$. %($\OPT(I_i) \geq m+i$, since $v_1$ and $v_2$
% receive $m+i$ requests in total, and $\OPT(I_i) \leq m+i$, since the
% following coloring is feasible: assign colors $1,\ldots,m$ to
% $v_1$, colors $m+1,\ldots,m+i$ to $v_2$, colors $1,\ldots,i$ to $v_3$,
% and colors $i+1,\ldots,m+i$ to $v_4$).
In order not to use more than $\OPT(I_i)$ colors for $I_i$, exactly $i$ of the colors assigned to $v_4$ have to be different from the colors assigned to $v_1$. The prefixes of length $2m$ in $S$ are identical, so all information to distinguish between the different sequences must be given as advice. The cardinality of $S$ is $m+1$.
To specify one out of $m+1$ possible actions,
$\CEIL{\log(m +1)}$ bits are necessary.
\end{proof}

For algorithms that are $\frac{9}{8}$-competitive or better, we give
the following lower bound.

\begin{theorem}\label{theorem:path-loweradvicebound}
 Consider multi-coloring paths of at least 10 nodes.
 For any $b \geq 3$ and any $(1+\frac{1}{2^b})$-competitive algorithm,
  \ALG, there exists an $N \in \nats$ such that \ALG 
  has advice complexity at least $b-2$ on sequences of length at least
  $N$.
\end{theorem}

\begin{proof}
For any $(1+\frac{1}{2^b})$-competitive algorithm, \ALG,
there exists an $\alpha \geq 1$ such that $\ALG(I) \leq (1+\frac{1}{2^b})
 \OPT(I) + \alpha$, for any input sequence $I$.
We consider sequences of length $n \geq 2^{2b+2} \alpha+3$.

Let $m = \FLOOR{\frac{n}{4}}$ and consider the same set of sequences
 as in the proof of Theorem~\ref{theorem:lowerbound_path}. 
Recall that $$\OPT(I_i) = m+i\,.$$
For the sequence $I_i$, let $x_i$ denote the
 number of colors that \ALG uses on $v_4$ but not on $v_1$.
Then, \ALG uses $m + x_i$ colors in total for $v_1$ and $v_4$.
On $v_3$, it can use at most $x_i$ of the colors used at $v_1$, so the
 total number of colors used at $v_1$, $v_2$, and $v_3$ is at 
 least $m+2i-x_i$.
Thus, %for the whole sequence, \ALG uses at least 
 $$\ALG(I_i) \geq \MAX{ m + x_i, m + 2i - x_i}\,.$$
% colors.

We will prove that there are $p \geq 2^{b-2}$ sequences
 $I_{i_1}, I_{i_2}, \ldots,I_{i_p}$ such that, for any pair $i_j \neq
 i_k$, we have $x_{i_j} \neq x_{j_k}$, or otherwise \ALG would not be 
 $(1+\frac{1}{2^b})$-competitive on sequences of at least $2^{b+2}$ requests.
%Since $p \geq 2$ by the assumption that $\eps < \frac15$,
% this will imply that \ALG needs at least $\FLOOR{\log p}+1$ advice bits.
This will immediately imply that \ALG must use at least $b-2$ advice bits.

Let $\eps = \frac{1}{2^b} + \frac{1}{2^{2b}}$.
From 
 $\ALG(I_i) \leq (1+\frac{1}{2^b}) \OPT(I_i) + \alpha$ and $m \geq
 2^{2b} \alpha$,
 we obtain the inequalities 
 $$m + x_i \leq (1+\eps)(m+i)$$ 
 and
 $$m + 2i - x_i \leq (1+\eps)(m+i)$$
 which reduce to
 \begin{equation}\label{eq:biggestx}
  x_i \leq \eps m + (1+\eps)i
 \end{equation}
 and
 \begin{equation}\label{eq:biggesti}
 i \leq \frac{x_i + \eps m}{1-\eps}
 \end{equation}
%\begin{equation}\label{eq:biggesti}
%  x_i \geq \frac{(2^b-1)i - m}{2^b} .
% \end{equation}

Hence, by (\ref{eq:biggestx}), $x_0 \leq \eps m$.
Thus, by  (\ref{eq:biggesti}), we can have $x_i=x_0$,
 only if $i \leq \frac{2\eps m}{1-\eps}$.
Therefore,  we let $i_1 = 0$ and $i_2 = \lfloor \frac{2\eps
  m}{1-\eps}+1 \rfloor$.
%Therefore, we let $i_1 = 0$ and let $i_2$ be the integer satisfying $\frac{2\eps m}{1-\eps} < i_2 \leq \frac{2\eps m}{1-\eps}+1$.
%%Therefore, let $i_1 = 0$ and $i_2 = \CEIL{\frac{2 m+1}{2^b-1}}$.
%Note that $i_j$ is an integer and therefore the upper bound might be slightly bigger than absolutely necessary.
In general, we ensure $x_{i_j} \neq x_{i_{j+1}}$ by letting $i_{j+1} =
\lfloor \frac{x_{i_j} + \eps m}{1-\eps}+1 \rfloor$.
Thus,
\begin{align*}
 i_{j+1}
 & \leq \frac{x_{i_j} + \eps m}{1-\eps}+1\\
 & \leq \frac{\eps m + (1+\eps)i_j + \eps m}{1-\eps} +1,
   \text{ by } (\ref{eq:biggestx})\\ 
 & = \frac{1+\eps}{1-\eps} \cdot i_j + \frac{2\eps m}{1-\eps} + 1
\end{align*}
Solving this recurrence relation, we get
\begin{align*}
 i_{j+1} 
 & \leq \left( \frac{1+\eps}{1-\eps} \right)^{j} \cdot i_1 + 
   \sum\limits_{k=0}^{j-1} 
     \left( \frac{1+\eps}{1-\eps} \right)^k 
     \left( \frac{2\eps m}{1-\eps} +1 \right)\\
 & = \left( \frac{1+\eps}{1-\eps} \right)^j \cdot 0 + 
     \frac{ \left( \frac{1+\eps}{1-\eps} \right)^j -1}
          { \frac{1+\eps}{1-\eps} -1}
     \left( \frac{2\eps m}{1-\eps} +1 \right)\\
 & = \frac{ \left( \frac{1+\eps}{1-\eps} \right)^j -1}
          {1+\eps-1+\eps}
     \left( 2\eps m + 1 - \eps \right)\\
 & = \frac{ \left( \frac{1+\eps}{1-\eps} \right)^j -1}
          {2\eps}
     \left( 2\eps m + 1 - \eps \right)
\end{align*}

We let $p$ equal the largest $j$ for which $i_{j} \leq m$:
%We let $p$ equal the smallest $j$ for which $i_{j+1} > m$:
\begin{align*}
 & m < i_{p+1} 
     \leq \frac{ \left( \frac{1+\eps}{1-\eps} \right)^p -1}
          {2\eps}
     \left( 2\eps m + 1 - \eps \right)\\
\Rightarrow ~ 
 & 2\eps m 
     < \left( \frac{1+\eps}{1-\eps} \right)^{p}
       \left( 2\eps m + 1 - \eps \right) -
       \left( 2\eps m + 1 - \eps \right)\\
\Leftrightarrow ~
& \frac{4m\eps + 1 - \eps}{2m\eps + 1 - \eps}
     < \left( \frac{1+\eps}{1-\eps} \right)^{p}\\
\Leftrightarrow ~
& \ln \left( 2 - \frac{1-\eps}{2m\eps + 1 - \eps} \right)
     < p \cdot \ln \left( 1 + \frac{2\eps}{1-\eps} \right)\\
\Rightarrow ~
& \ln \left( 2 - \frac{1-\eps}{2m\eps + 1 - \eps} \right)
     < p \cdot \frac{2\eps}{1-\eps}, 
       \text{ since } \ln(1+x) \leq x, \text{ for } x > -1
       \\
\Rightarrow ~
& \ln \left( 2 - \frac13 \right) 
     < p \cdot \frac{2\eps}{1-\eps},
       \text{ since } m\eps > 2^b > 1-\eps, 
%       \text{ because } n > 2^{b+2} = \frac{4}{\eps}
       \\
\Rightarrow ~
& \ln \left( \sqrt{e} \right) 
     < p \cdot \frac{2\eps}{1-\eps}\\
\Leftrightarrow ~
& \frac12 
     < p \cdot \frac{2\eps}{1-\eps}\\
\Leftrightarrow ~
& p > \frac{1-\eps}{4\eps}\\
\Leftrightarrow ~ 
& p > \frac{1-\frac{1}{2^b}-\frac{1}{2^{2b}}}{\frac{4}{2^b}+\frac{4}{2^{2b}}} 
    = \frac{2^{2b}-2^b-1}{2^{b+2}+4} 
    > 2^{b-2}-1\\
\Rightarrow ~
& p \geq 2^{b-2},
    \text{ since } p \text{ is an integer}
\end{align*}
This completes the proof.
\end{proof}

%Theorems~\ref{theorem:lowerbound_path} and
% \ref{theorem:path-loweradvicebound} give lower bounds on the amount
% of advice needed to obtain 
%certain competitive ratios.
%We now show that to be $\frac43$-competitive, $\omega(1)$ advice bits
%are needed.

For the following theorem, we define the class of {\em \local}
 algorithms:
A multi-coloring algorithm, \ALG, is called \local, if there exists a
 constant $d$ such that,
 when assigning a color to a request to a node $v$, \ALG
 bases its decision only on requests to nodes a distance of at most
 $d$ away from $v$.
Note that, in particular, a \local algorithm cannot base its decision
 on the current value of \OPT.

\begin{theorem}
\label{theorem:fourthirds}
 No \local online algorithm for multi-coloring paths with advice complexity $O(1)$ can be better than $\frac{4}{3}$-competitive.
\end{theorem}

\begin{proof}
 Having an online algorithm with advice complexity $O(1)$ gives an algorithm a constant number of possible algorithmic behaviors; it is equivalent to having $O(1)$ online algorithms without advice
and choosing one of these according to the given advice.

As shown in~\cite{CFL13}, the family of sequences used in the proofs
 of Theorems~\ref{theorem:lowerbound_path}
 and~\ref{theorem:path-loweradvicebound}, can be used to prove that any
 online algorithm without advice has a competitive ratio of at least
 $\frac43$. 
The result is asymptotic, since the construction works with any
 scaling of the number of requests to each node.
This means that for each algorithm, there exists an infinite family of
sequences indexed by $n$, the length of the sequences, establishing the
lower bound for each algorithm.

For any \local algorithm \ALG, there is a constant $d$ such that, when
 assigning a color to a request, \ALG ignores all requests given to
 nodes a distance of more than $d$ away from the requested node.
For any $n$ and any family, there is a smallest and a largest node on
the path which is requested, and the part of the path from this
smallest to the largest node defines a subpath.
We now rename nodes in these infinite families so that the
subpaths used by the different families are separated by $d$
unused nodes. We then form one request sequence by concatenating all these
renamed subsequences. We scale the number of requests in each sequence
such that the value of \OPT is the same for each sequence.

Clearly, no matter which of the $O(1)$ algorithms are run on this
constructed family, its performance tends to at least $\frac{4}{3}\OPT$.
\end{proof}

\subsection{Upper bounds}
\label{subsec:pathUpper}
For multi-coloring of a path, there exists a strictly $1$-competitive
$1$-recoloring algorithm, \GEO~\cite{CFL13}. \GEO divides the nodes into two sets, \UPPER and \LOWER, such
 that every second node belongs to \UPPER and the remaining nodes
 belong to \LOWER.
The following invariant is maintained:
 After each request,
 each node in \LOWER uses consecutive colors starting with the color 1 and
 each node in \UPPER uses consecutive colors ending with a color no
  larger than the optimal number of colors for the
  sequence of requests seen so far.

The algorithm for paths from~\cite{CFL13} is easily generalized to
work on bipartite graphs, letting the nodes of one partition, $L$, belong
to \LOWER and the nodes of the other partition, $U$, belong to \UPPER.
Recoloring is only needed if the number of colors used by an optimal offline algorithm is not known.
Hence, using $\ENCODE{\OPT}$ advice bits, an online algorithm can be
strictly $1$-competitive, even if recoloring is not allowed.
We call the resulting algorithm \GEOA.
%We design \GEOA, a variant of the algorithm \GEO, to be used for
%multi-coloring of a bipartite graph, given the graph partition into
%two independent sets, $L$ and $U$, and as advice the number of colors
%\OPT uses.

To describe the algorithm \GEOA in detail, we need some notation:
Let $f_i(v)$ denote the set of colors assigned to node~$v$ after
the first $i$ requests, starting with request~1.
Also, for notational convenience,
we define $f_0(v)=\emptyset$ for all~$v$.
To smoothly handle initially empty sets of colors in the algorithm,
we define that if $f_i(v)$ is the empty set,
then $\min f_i(v) = \max f_i(v)= 0$.
This notation will be used throughout the appendix.
\GEOA is listed as Algorithm~\ref{algorithm-optimal-advice-bipartite}.

\begin{algorithm}[thb]
\algsetup{indent=2em}
\begin{algorithmic}[1]
\STATE Assume that a bipartite graph is given by the partition into $L$ and $U$.
%%\STATE \textbf{advice $\advice$} /* read $\ENCODE{\OPT}$ bits of advice */
%\STATE $\maxcolor = \textbf{advice}(\ENCODE{\OPT})$ /* read $\ENCODE{\OPT}$ bits of advice */
\STATE {\bf Advice:} $m = \OPT$
\FOR{ $i = 1$ \TO $n$ }
\STATE Assume that the $i$th request, $r$, is to node~$v$
   \IF{$v \in$ $U$}
   \STATE /* using the upper colors top-down */
      \IF{$f_{i-1}(v) = \emptyset$}
         \STATE{give $r$ color $\maxcolor$}
      \ELSE
         \STATE give $r$ color $\min f_{i-1}(v) - 1 $
      \ENDIF
   \ELSE
      \STATE /* $v \in L$;  using the lower colors bottom-up */
      \STATE{give $r$ color $\max f_{i-1}(v) + 1$}
   \ENDIF
\ENDFOR
\end{algorithmic}
\caption{The multi-coloring algorithm \GEOA.}
\label{algorithm-optimal-advice-bipartite}
\end{algorithm}

\begin{theorem}\label{theorem:bipartite-upperbound-algorithm}
Algorithm \GEOA is correct, strictly $1$-competitive, and
has advice complexity $\ENCODE{\OPT}$.%, where $\OPT \leq n$.
\end{theorem}

\begin{proof}
We consider correctness first.
%As advice, the algorithm asks for the number of colors that \OPT uses,
%and then maintains a variant of the online algorithm, \GEO~\cite{CFL13},
%defined as follows:
% Given the partitions $L$ and $U$ of a bipartite graph and knowing the number of colors that \OPT uses, \GEOA assigns the nodes consecutive colors from $1$ to $\OPT$ starting from the bottom in $L$ and from the top in $U$.
Clearly, at time $i$, the maximum color assigned to a node $v\in L$ is
$\max f_i(v) = |f_i(v)|$ and the minimum color assigned to a node
$v\in U$ is $\min f_i(v) = \OPT + 1 - |f_i(v)|$ (assuming $v$ has
received at least one request).

 Assume for the sake of contradiction that, at some time $i$, a request to a node $l\in L$ gets assigned the same color $c$ as a request to a neighboring node $u$,
which must belong to $U$.
This means that $c = |f_i(l)|$ and $c = \OPT + 1 - |f_i(u)|$, and, as $l$ and $u$ are neighbors, $\OPT \geq |f_i(l)| + |f_i(u)|$, but then $\OPT \geq |f_i(l)| + |f_i(u)| =  c + \OPT + 1 - c = \OPT + 1$. This is a contradiction, so \GEOA is correct.

It follows directly that the maximum color that \GEOA assigns is \OPT,
implying that \GEOA is strictly $1$-competitive.

The maximum color that an optimal offline algorithm uses, given a sequence of length $n$, is
$n$. Therefore, $\OPT \leq n $ and $\ENCODE{\OPT}$ advice bits are sufficient.
\end{proof}

We turn to nonoptimal variants of \GEOA using fewer than
 $\ENCODE{\OPT}$ advice bits.
We show how to obtain
 a particular competitive ratio
of $1+\frac{1}{2^{b}}$, using $b+1+O(\log\log\OPT)$ bits of advice.
Thus, essentially, we are approaching optimality exponentially
fast in the number of bits of advice.

\begin{theorem}\label{theorem:oneplus-competitivealgorithm} For any
  integer $b\geq 1$, there exists a strictly
  $(1+\frac{1}{2^{b-1}})$-competitive online algorithm for
  multi-coloring bipartite graphs with advice complexity
  $b + \ENCODE{a}$, where $a+b$ is the total number of bits in
  the value \OPT.%, i.e., $a = \CEIL{ \log(\OPT +1) } - b$. 
\end{theorem}

\begin{proof}
As advice, the algorithm asks for the $b$ high order bits of the value \OPT,
as well as the number $a = \CEIL{ \log(\OPT+1) } - b$ of 
low order bits, but not the value of these bits.
The algorithm knows $b$ and can therefore just read the first $b$ bits,
while $a$ needs to be encoded. Thus, $b + \ENCODE{a}$ bits are sufficient
to encode the advice.

First, if \OPT contains fewer than $b$ bits, this is detected by
$a$ being zero. In this case, some of the $b$ bits may be leading zeros.
By Theorem~\ref{theorem:bipartite-upperbound-algorithm},
we can then be strictly $1$-competitive.

Now, assume this is not the case. Let $\OPT_b = \FLOOR{ \frac{\OPT}{2^a}
}$ denote the
value represented by the $b$ high order bits.
Then the algorithm computes
$\maxcolor = 2^a \OPT_b + 2^a -1$ and runs
%and clearly $m=\FLOOR{\frac{\OPT}{2^y}}$
%By the proof of Lemma~\ref{lemma:one-comp-algorithm}, running
\GEOA with this $m$. % (without asking for further advice), % and $d=2^{a+1}$,
Since $\OPT \leq m \leq \OPT + 2^a -1$, the algorithm is correct and
uses at most $\OPT + 2^{a} - 1$ colors.

For any number $x \geq 1$, consisting of $c$ bits, 
with the most significant bit being one,
$2^{c} \leq 2x$.
Thus, $2^{b+a} \leq 2\OPT$, so $2^{a} \leq \frac{2\OPT}{2^b}$.
This means that the number of colors used by \GEOA is less than
$\OPT + \frac{2\OPT}{2^b} = (1 + \frac{1}{2^{b-1}})\OPT$, so
the algorithm is strictly $(1 + \frac{1}{2^{b-1}})$-competitive.
\end{proof}

Considering the lower bound of Theorem~\ref{theorem:lowerbound_path}
versus the upper bound of
Theorem~\ref{theorem:bipartite-upperbound-algorithm}, as well as the
lower bound of Theorem~\ref{theorem:path-loweradvicebound}
versus the upper bound of Theorem~\ref{theorem:oneplus-competitivealgorithm},
in both cases there is a small discrepancy of a few bits,
in addition to a low order term.
The lower bound proof of Theorem~\ref{theorem:lowerbound_path}
demonstrates the need of advice to distinguish between
$\FLOOR{\frac{n}{4}}+1$ different scenarios to be optimal.
It will vary with $n$ whether or not the division by four saves
one or two bits compared with $\log n$, and similar reasoning applies
to Theorem~\ref{theorem:path-loweradvicebound}.
Thus, when stating the
lower bound, we have to subtract two bits
(refer to Table~\ref{table-results}).
Using encoding tricks, to for instance identify cases where \OPT
has a very small value, we can also sometimes get down to a bit less
than $\log n$ for the upper bound.
Thus, our results are nearly tight, up to low order terms,
but because of rounding,
it seems difficult to squeeze the missing few bits out of the bounds
in every case. Note that for upper bounds, one could
perform better by distinguishing between different cases, but
finding out which case to use requires extra bits, by which we lose the
advantage again.

\begin{corollary}
\label{theorem:oneplusepsilon}
For any $\varepsilon > 0$, there exists a strictly
$(1+\varepsilon)$-competitive deterministic online algorithm for
multi-coloring bipartite graphs with advice complexity
\[O(\log\log\OPT).\]
%\[O(\log\log\OPT)\subseteq O(\log\log n).\]
\end{corollary}

\begin{proof}
Except for the term $b$, the advice stated in
Theorem~\ref{theorem:oneplus-competitivealgorithm} is
$O(\log\log\OPT)$ and $\OPT\leq n$.
Thus, we just need to bound the term $b$.
For a given \eps, choose $b$ large enough such that $\frac{1}{2^{b-1}} \leq
 \varepsilon$.
Using this value for $b$ in Theorem~\ref{theorem:oneplus-competitivealgorithm},
 we obtain an algorithm with a strict competitive ratio of at most $1
 + \frac{1}{2^{b-1}} \leq 1 + \eps$.
%
%For a given \eps, choose a value $V_{\OPT}$ of the form $2^{2^x}$, for
%some positive integer $x$, and large enough that $\frac{2}{\log
% V_{\OPT}} \leq \varepsilon$.
%Using $b=\log\log V_{\OPT}$, we obtain a strictly
%$(1+\frac{2}{2^b})$-competitive algorithm,
%and since
%$\frac{2}{2^b}=\frac{2}{2^{\log\log V_{\OPT}}}
%              =\frac{2}{\log V_{\OPT}}
%              \leq \varepsilon$,
%the competitiveness follows.
%
Since, for any given \eps, $b$ is a constant, the total amount of
 advice is $O(\log\log\OPT)$.
\end{proof}

%As a remark, in the above result, we can trade strictness in for advice,
%i.e., if we allow for an additive constant, switching to asymptotic
%competitiveness, we can decrease the number of advice bits used.
%The advice bits, other than $b$, can be encoded using $\log\log\OPT$
%bits, plus some low order terms. We can also restrict $b$ to
%$\log\log\OPT$ bits as follows:
%%
%Let $\alpha$ be the maximal number of colors
%used by the algorithm, when receiving only
%$2\log\log\OPT(I) + o(\log\log\OPT(I))$ bits of
%advice, over all sequences $I$ where $\OPT(I)\leq V_{\OPT}$.
%Since $V_{\OPT}$ is a constant only dependent on $\varepsilon$, it follows
%that $\alpha$ is also a constant, so the algorithm's maximum color used
%on any sequence $I$ is at most
%$(1+\varepsilon)\OPT(I)+\alpha$, and the algorithm is asymptotically
%$(1+\varepsilon)$-competitive.

\subsection{Cancellations}
\begin{algorithm}[thb]
\algsetup{indent=2em}
\begin{algorithmic}[1]
\STATE Assume that a bipartite graph is given by the partition into $L$ and $U$.
%\STATE \textbf{ADVICE $\advice$} read $\ENCODE{\OPT}$ bits of Advice
%\STATE Let $\maxcolor = \advice$ be the maximum color
\STATE {\bf Advice:} $m = \OPT$
\FOR{ $i = 1$ \TO $n$ }
\STATE Assume that the $i$th request, $r$, is to node~$v$
   \IF{$r$ is a color request}
      \IF{$v \in$ $U$}
      \STATE /* using the upper colors top-down */
         \IF{$f_{i-1}(v) = \emptyset$}
            \STATE{give $r$ color $\maxcolor$}
         \ELSE
            \STATE give $r$ color $\min f_{i-1}(v) - 1 $
         \ENDIF
      \ELSE
         \STATE /* $v \in L$; using the lower colors bottom-up */
         \STATE{give $r$ color $\max f_{i-1}(v) + 1 $}
      \ENDIF
   \ELSE
   \STATE /* $r$ is a cancellation */
      \IF{$v \in U$}
         \IF{the color of $r$ is different from $\min f_{i-1}(v)$}
             \STATE recolor the request that has color $\min f_{i-1}(v)$, giving it the color of $r$
         \ENDIF
      \ELSE
         \IF{the color of $r$ is different from $\max f_{i-1}(v)$}
             \STATE recolor the request that has color $\max f_{i-1}(v)$, giving it the color of $r$
         \ENDIF
      \ENDIF
   \ENDIF
\ENDFOR
\end{algorithmic}
\caption{The multi-coloring algorithm \GEOAC.}
\label{algo:advice-cancellation}
\end{algorithm}

%\begin{theorem}
% Using $0$-recoloring \GEOA can be used to multi-color a bipartite graph optimally (Algorithm~\ref{algo:advice-cancellation}).
%\end{theorem}
%
%\begin{proof}
% With the $0$-recoloring, we assure that even with cancellations each node only uses consecutive colors, implying optimality.
%\end{proof}

The Multi-Coloring problem is sometimes considered
in the context of request cancellations,
i.e., a color already given to a node disappears again.
We observe that even using the
weakest form of recoloring, namely $0$-recoloring, where only requests
at the node where the cancellation takes place may be recolored, we
can extend the algorithm \GEOA, using the same advice, to a strictly $1$-competitive
algorithm. This is simply done by recoloring at most one request per cancellation to
ensure that the invariants regarding lower and upper nodes are maintained,
i.e., ensuring that the colors used at any node form a consecutive sequence
starting from one and increasing and starting from \OPT and decreasing for lower
and upper nodes, respectively.
This algorithm, \GEOAC, is listed as
 Algorithm~\ref{algo:advice-cancellation}.
Note that the difference to
 Algorithm~\ref{algorithm-optimal-advice-bipartite} is the check in
 line 5 as to whether the current request is a color request and the
 addition of lines 15--22 handling cancellations.

\section{Hexagonal Graphs}
A hexagonal graph is a graph that can be obtained by placing (at most)
 one node in each cell of a hexagonal grid (such as the one sketched in
 Figure~\ref{fig:advice-pathlong-inmain}) and adding an edge between
 any pair of nodes placed in neighboring cells.
Note that any hexagonal graph can be $3$-colored. This is easily seen,
 since it is possible to use the three colors cyclically on the cells
 of each row of the underlying hexagonal grid, such that no two neighboring cells
 receive the same color.

\subsection{Lower bounds}

%We omit the formal definition of hexagonal graphs, but they are graphs,
%where each vertex has up to six neighbors, such that the graphs could be layed
%out in a regular pattern, as shown in the example of
%Figure~\ref{fig:advice-pathlong-inmain}.
%Hexagonal graphs can be $3$-colored.
%This is easily seen by overlaying the graph with a graph without
%holes in any row. Color the top row using the colors cyclicly.
%Then, recursively color each neighboring row. Unless the row above had only
%one cell, there will be exactly one way to legally color the current row.

%As in the previous section, we first focus on lower bound results.

\begin{theorem}\label{theorem:optimal-hex}
Any online algorithm for multi-coloring hexagonal graphs with a
strict competitive ratio strictly smaller than $\frac{3}{2}$ has advice
complexity at least $\FLOOR{\frac{n-1}{3}}$.
\end{theorem}
\begin{proof}
First, we explain a small part of the construction that we will use in many copies.
We consider two sequences with the same prefix of length
$2$.
Both sequences can be colored with two colors, but this requires
coloring the two prefixes of length two differently.
%Referring to Figure~\ref{fig:advice-pathlong-inmain}, using only the
%left-most part,
Consider the left-most part of Figure~\ref{fig:advice-pathlong-inmain}
 (surrounded by thick lines)
 consisting of the ``double'' nodes $D_1$ and $ D_2$, the ``outer''
 nodes $O_0$ and $O_1$ and the ``single'' nodes $ S_1$, and $ S_2$. 
These nodes form the same type of
configuration as the nodes $D_3$, $D_4$, $O_1$, $O_2$, $S_3$, and $S_4$.
%The first group of nodes has two ``outer'' nodes $O_0$ and $O_1$. 
%If $O_0$ and $O_1$
% are given some requests, they can later be ``connected'' 
% by follow-up requests to either
% $D_1$ and $D_2$ or to $S_1$. 
If a pair of outer nodes
 are given some requests, they can later be ``connected'' 
 by follow-up requests to either
 the two double nodes or the single node between them.

First the nodes $O_0$ and $O_1$ get one request each. 
Then, either $D_1$ and $D_2$ or $S_1$ and $S_2$ receive one request
 each.
The node $S_2$ is used to get up to the same sequence length
in all cases.
In order not to use more than two colors,
the outer nodes have to use
different colors if we later give requests to the two $D$-nodes.
Similarly, the $O$-nodes should have the same color if we later
give a request to the $S$-node in between them.
Since the prefix of length two is
$\SEQ{O_0, O_1}$ for both sequences, all information for an algorithm
to distinguish between the two sequences must be given as advice.
%The size of $S$ is $2$.

We can repeat this graph pattern $\FLOOR{ \frac{n-1}{3} }$ times,
as illustrated in Figure~\ref{fig:advice-pathlong-inmain} with $k =
 \FLOOR{ \frac{n-1}{3} }$, 
giving the requests to all $O$-nodes first.
\begin{figure}[t]
\begin{center}
\resizebox{9.5cm}{!}{
\begin{tikzpicture} [thick, scale = 0.6, hexa/.style= {shape=regular polygon,regular polygon sides=6,minimum size=1.2cm, draw,inner sep=0,anchor=south,fill=white,rotate=30},  hl/.style={line width=2pt, line cap=round} ]
 \node[hexa] (0;0) at ({0.000000*sin(60)},{0.000000}) {}; 
 \node[hexa] (1;0) at ({2.000000*sin(60)},{0.000000}) {}; 
 \node[hexa] (2;0) at ({4.000000*sin(60)},{0.000000}) {}; 
 \node[hexa] (3;0) at ({6.000000*sin(60)},{0.000000}) {}; 
 \node[hexa] (4;0) at ({8.000000*sin(60)},{0.000000}) {}; 
 \node[hexa] (6;0) at ({12.000000*sin(60)},{0.000000}) {}; 
 \node[hexa] (7;0) at ({14.000000*sin(60)},{0.000000}) {}; 
 \node[hexa] (8;0) at ({16.000000*sin(60)},{0.000000}) {}; 
 \node[hexa] (9;0) at ({18.000000*sin(60)},{0.000000}) {}; 
 \node[hexa] (0;1) at ({1.000000*sin(60)},{1.500000}) {}; 
 \node[hexa] (1;1) at ({3.000000*sin(60)},{1.500000}) {}; 
 \node[hexa] (2;1) at ({5.000000*sin(60)},{1.500000}) {}; 
 \node[hexa] (3;1) at ({7.000000*sin(60)},{1.500000}) {}; 
 \node[hexa] (5;1) at ({11.000000*sin(60)},{1.500000}) {}; 
 \node[hexa] (6;1) at ({13.000000*sin(60)},{1.500000}) {}; 
 \node[hexa] (1;-1) at ({1.000000*sin(60)},{-1.500000}) {}; 
 \node[hexa] (2;-1) at ({3.000000*sin(60)},{-1.500000}) {}; 
 \node[hexa] (3;-1) at ({5.000000*sin(60)},{-1.500000}) {}; 
 \node[hexa] (4;-1) at ({7.000000*sin(60)},{-1.500000}) {}; 
 \node[hexa] (6;-1) at ({11.000000*sin(60)},{-1.500000}) {}; 
 \node[hexa] (7;-1) at ({13.000000*sin(60)},{-1.500000}) {}; 
 \node[hexa] (2;-2) at ({2.000000*sin(60)},{-3.000000}) {}; 
 \node[hexa] (6;-2) at ({10.000000*sin(60)},{-3.000000}) {}; 
 \node[hexa] (7;-2) at ({12.000000*sin(60)},{-3.000000}) {}; 
 \node[hexa] (2;2) at ({6.000000*sin(60)},{3.000000}) {}; 
 \node[circle,Black,minimum size=2cm] at (0;0) {$O_0$};
 \node[circle,Black,minimum size=2cm] at (1;0) {$S_1$};
 \node[circle,Black,minimum size=2cm] at (2;0) {$O_1$};
 \node[circle,Black,minimum size=2cm] at (3;0) {$S_3$};
 \node[circle,Black,minimum size=2cm] at (4;0) {$O_2$};
 \node[circle,Black,minimum size=2cm] at (6;0) {$S_{2k-1}$};
 \node[circle,Black,minimum size=2cm] at (7;0) {$O_{k}$};
 \node[circle,Black,minimum size=2cm] at (0;1) {$D_1$};
 \node[circle,Black,minimum size=2cm] at (1;1) {$D_2$};
 \node[circle,Black,minimum size=2cm] at (3;-1) {$D_3$};
 \node[circle,Black,minimum size=2cm] at (4;-1) {$D_4$};
 \node[circle,Black,minimum size=2cm] at (5;1) {$D_{2k-1}$};
 \node[circle,Black,minimum size=2cm] at (6;1) {$D_{2k}$};
 \node[circle,Black,minimum size=2cm] at (2;-2) {$S_2$};
 \node[circle,Black,minimum size=2cm] at (2;2) {$S_4$};
 \node[circle,Black,minimum size=2cm] at (7;-2) {$S_{2k}$};
 \node[circle,Black,minimum size=2cm] at (9;0) {$R$};
 \node[circle,Black,minimum size=2cm] at ({9.500000*sin(60)},{0.500000}) {$\dots$}; 
 %corners are labelled
 %    1
 %  2   6
 %  3   5
 %    4
 \draw[hl]  (2;-2.corner 1) -- (2;-2.corner 2)  -- (2;-2.corner 3)  -- (2;-2.corner 4)  -- (2;-2.corner 5)  -- (2;-2.corner 6) -- (2;-2.corner 1);
 \draw[hl]  (0;1.corner 3) -- (0;1.corner 2)  -- (0;1.corner 1) -- (0;1.corner 6) -- (1;1.corner 1)  -- (1;1.corner 6) -- (1;1.corner 5) -- (2;0.corner 6) -- (2;0.corner 5) -- (2;0.corner 4) -- (2;0.corner 3) -- (1;0.corner 4) -- (1;0.corner 3) -- (0;0.corner 4) -- (0;0.corner 3)-- (0;0.corner 2) -- (0;0.corner 1);
\end{tikzpicture}}
\end{center}
\caption{Hexagonal lower bound construction.}
\label{fig:advice-pathlong-inmain}
\end{figure}
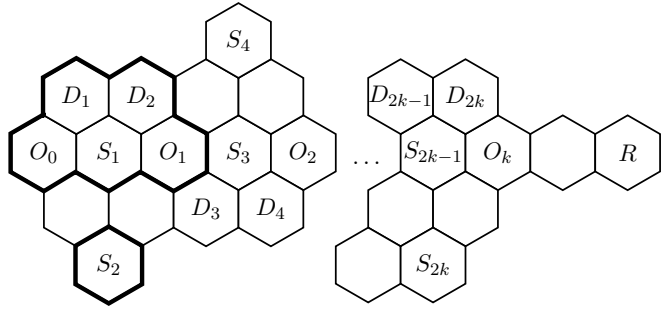

We now define the set of sequences $S$ of cardinality
$ 2^{\FLOOR{ \frac{n-1}{3} }}$ formally, i.e., we define a sequence for each 
possible combination of requests to either $D_{2j-1}$ and $D_{2j}$ or $S_{2j-1}$ and $S_{2j}$ for 
$j=1,2,\dots,\FLOOR{ \frac{n-1}{3} }$ .
A sequence is defined for any chosen combination of the following
$i$-values, i.e., by choosing a tuple
$(i_1,i_2,\dots,i_{\FLOOR{ \frac{n-1}{3} }}) \in \SET{0,1}^{\FLOOR{ \frac{n-1}{3} }}$.
For any such choice, we define the sequence as a concatenation of the
subsequences given below.
In the description of the subsequences, we use the notation $\REQ{v}{m}$ to denote a sequence of $m$ requests to a node $v$,
and also use this notation for $m=0$, denoting the empty request sequence, and $m=1$, denoting
one request.

\begin{itemize}
% \item $\REQ{O_0}{1}$
 \item $\REQ{O_{j}}{1}$ for $j = 0,1,\dots,\FLOOR{ \frac{n-1}{3} }$
 \item $\REQ{D_{2j-1}}{i_j}$, $\REQ{D_{2j}}{i_j}$ for $j = 1,2,\dots,\FLOOR{ \frac{n-1}{3} }$
 \item $\REQ{S_{2j-1}}{1-i_j}$, $\REQ{S_{2j}}{1-i_j}$ for $j = 1,2,\dots,\FLOOR{ \frac{n-1}{3} }$
 \item $\REQ{R}{n - (3\FLOOR{\frac{n-1}{3}}+1)}$
\end{itemize}
Note that for any $i_j$, $\SET{i_j, 1-i_j}=\SET{0,1}$ and we either give requests to the $D$-nodes or the $S$-nodes.
The possible requests to $R$ simply gets all sequences up
to a length of $n$.

The node $O_0$ is given some color. After that, we have
$\FLOOR{\frac{n-1}{3}}$ independent choices of coloring each node
$O_i$ in the prefix of any sequence identically to $O_{i-1}$ or not.
Since the prefixes are the same,
all information for an algorithm to distinguish between the different
sequences must be given as advice.
To specify one out of $2^{\FLOOR{\frac{n-1}{3}}}$ possible actions,
$\CEIL{ \log 2^{ \FLOOR{ \frac{n-1}{3} } } } = \FLOOR{\frac{n-1}{3}}$
bits are necessary.
\end{proof}

\begin{theorem}
\label{theorem:hexfivefour}
Any online algorithm for multi-coloring hexagonal graphs with
competitive ratio strictly smaller than $\frac{5}{4}$ has advice complexity
$\Omega(n)$.
\end{theorem}

\begin{proof}
We use the basic construction from Theorem~\ref{theorem:optimal-hex}.
Assume $p$ requests are given to one of the components like this:

First, we give $\frac{p}{4}$ requests to each of $O_0$ and $O_1$.
Let $q$, $0 \leq q \leq \frac{p}{4}$, denote the number of colors used at both nodes. Then following up by
giving $\frac{p}{4}$ requests to each $S$-node results in a minimum
of $\frac{3p}{4}-q$ colors used, while giving the requests to the
$D$-nodes instead results in a minimum of $\frac{p}{2}+q$ colors.
%Any online algorithm must decide on the value of $q$, and, to get the best
%worst-case result, must minimize the maximum of $\frac{3p}{4}-q$ and
%$\frac{p}{2}+q$. 
%This is obtained by choosing $q=\frac{p}{8}$, giving a result of $\frac{5p}{8}$.

Note that $\OPT = \frac{p}{2}$, independent of in which of the two ways the
sequence is continued.
Thus, for any $\eps > 0$, any $(\frac54-\eps)$-competitive algorithm
 must choose $q$ such that, for some constant $\alpha$, $\frac{3p}{4}-q
 \leq \left(\frac54-\eps\right)\frac{p}{2} +\alpha$ and $\frac{p}{2} +q
 \leq \left(\frac54-\eps\right)\frac{p}{2} +\alpha$.
%\begin{align*}
%  & \frac{3p}{4}-q \leq \left(\frac54-\eps\right)\frac{p}{2} +b \;
%    \text{ and }\\ 
%  & \frac{p}{2} +q \leq \left(\frac54-\eps\right)\frac{p}{2} +b
%\end{align*}
Adding these two inequalities, we obtain $\frac{5p}{4} \leq
 (\frac54-\eps)p+2\alpha$ which is equivalent to $\eps p \leq 2\alpha$.
Thus, if $p$ is non-constant, no $(\frac54-\eps)$-competitive
 algorithm can use the same value of $q$ for both sequences.

Now assume for the sake of contradiction that for some advice of $g(n)
 \in o(n)$ bits, we can obtain a ratio of $\frac54-\eps$.
Let $f(n) = \frac{1}{2} \frac{n}{g(n)}$.
Since $g(n) \in o(n)$, $f(n) \in \omega(1)$.
The idea is now to repeat the construction as in
 the proof of Theorem~\ref{theorem:optimal-hex} and give $f(n)$
 requests to each construction ($f(n)$ has the role of $p$ in the
 above).
Since a pair of neighboring constructions share $f(n)/4$ requests,
this results in $\frac{n-f(n)/4}{3f(n)/4} = \frac{4n-f(n)}{3f(n)} \geq
\frac{n}{f(n)}$ 
constructions. We assume 
without loss of generality 
that all our divisions result in integers.

In order to be $(\frac{5}{4}-\eps)$-competitive, an online
algorithm must, for each two neighboring $O$-nodes, choose between at
 least two different values of $q$.
These are independent decisions, and the ratio only ends up strictly better
than $\frac{5}{4}$ if the algorithm decides correctly in every subconstruction.
%whether the later requests are to $S$-nodes or $D$-nodes.
Thus, it needs at least $\frac{n}{f(n)}$ bits of advice.
However, $\frac{n}{f(n)}=\frac{n}{\frac{1}{2} \frac{n}{g(n)}}=2g(n)>g(n)$,
which is a contradiction.
%
%Note that since $f(n) \in \omega(1)$, an erroneous decision cannot
%be controlled by the additive constant in the definition of the
%competitive ratio.
\end{proof}

\subsection{Upper bounds}

We have the following trivial upper bound on the advice necessary to
be optimal, independent of the graph topology:

\begin{theorem}\label{theorem:optimal-hex-upper}
 There is a strictly $1$-competitive online multi-coloring algorithm 
 with advice complexity
 $(n+1)\CEIL{\log \OPT}$.
\end{theorem}

\begin{proof}
Start by asking for the number of bits necessary to represent values
up to \OPT.
Then for each request, read $\CEIL{\log(\OPT+1)}$ bits telling, which color
to use.
This gives
$\ENCODE{\CEIL{\log \OPT}}+n\CEIL{\log \OPT} <
 (n+1)\CEIL{\log \OPT}$.
\end{proof}

In the following, we will show how two known approximation
 algorithms can be converted to online algorithms with advice.
In the description of the algorithms, we let the {\em weight} of a
 clique denote the total number of requests to the nodes of the
 clique.
Note that the only maximal cliques in a hexagonal graph are isolated
nodes, edges, or triangles.
We let $\omega$ denote the maximum weight of any clique in the
 graph.\footnote{The Greek letter $\omega$ is traditionally used here,
 so we will also do that. 
Since there is no argument, this should not give rise to confusion with
 the $\omega(f)$, stemming from asymptotic notation.} 

A $\frac32$-competitive algorithm called the Fixed Preference Allocation algorithm, \FPA,
 was proposed in~\cite{JKM99}. 
%A ratio of $\frac{3}{2}$ was established using a $1$-recoloring algorithm.
In~\cite{N02}, the strategy was simplified and it was noted that the
 algorithm can be converted to a $1$-recoloring online algorithm.
We describe the simplified offline algorithm below.

The algorithm uses three color classes, \classred, \classgreen, and
 \classblue.
%For a node $v$, $\class(v)$ is one of these three options.
The color classes represent a partitioning of the nodes in the graph
so that no two neighbors are in the same partition.
Each of the three color classes has its own set of
$\CEIL{\frac{\omega}{2}}$ colors, and each node in a given
color class uses the colors of its color class, starting with the smallest.
%For a node $v$, we let $\colors(\class(v))$ denote 
This set of colors is also referred to as the node's {\em private} colors.
If more than $\CEIL{\frac{\omega}{2}}$ requests are given to a node,
then it borrows colors from the private colors of one of its neighbors,
taking the highest available color.
\classred nodes can borrow colors from \classgreen nodes, \classgreen
 from \classblue, and \classblue from \classred.

For completeness, we give the arguments 
 that \FPA is correct and obtains an approximation ratio of~$\frac32$.
Assume for the purpose of contradiction that the coloring produced by
the algorithm causes a conflict between an \classred node and a \classgreen 
node. This means that their combined number of requests must be
greater than $\omega$, which is a contradiction. The same argument
holds for the other color combinations. Thus, the coloring is legal. 
Any optimal algorithm needs at least $\omega$ colors, so $\OPT \geq
\omega$ and the algorithm is a $\frac{3}{2}$-approximation algorithm. 

\begin{algorithm}[thb]
\algsetup{indent=2em}
\begin{algorithmic}[1]
% \STATE \textbf{ADVICE $\advice$} read $\ENCODE{\frac{n}{2}}$ bits of
% Advice
 \STATE {\bf Advice:} $\CEIL{\frac{\omega}{2}}$
% \STATE $\CEIL{\frac{\omega}{2}} = \advice$
 \STATE $\RED = \SET{ 1,2,\dots,\CEIL{\frac{\omega}{2}}}$, 
 \STATE $\GREEN = \SET{ \CEIL{\frac{\omega}{2}}+1,\CEIL{\frac{\omega}{2}}+2,\dots,2\CEIL{ \frac{\omega}{2} }}$, 
 \STATE $\BLUE = \SET{ 2\CEIL{\frac{\omega}{2}}+1,2\CEIL{\frac{\omega}{2}}+2,\dots,3\CEIL{\frac{\omega}{2}}}$
 \STATE {\bf Function} $\class(v)$
 \STATE \qquad {\bf return} $v$'s color class: \classred, \classgreen, or \classblue
 \STATE {\bf Function} $\borrow(c)$
 \STATE \qquad {\bf return} the next class in the wrap-around sequence \classred, \classgreen, or \classblue
 \STATE {\bf Function} $\colors(c)$
 \STATE \qquad {\bf return} the set of private colors of class~$c$
 \FOR{ $i = 1$ \TO $n$ }
\STATE Assume that the $i$th request, $r$, is to node~$v$
  \IF{$\SIZE{ f_{i-1}(v)} < \CEIL{\frac{\omega}{2}}$ } %color available
     \STATE give $r$ color $\min(\colors(\class(v))\setminus f_{i-1}(v))$
  \ELSE %Borrow a color
     \STATE give $r$ color $\max(\colors(\borrow(\class(v)))\setminus f_{i-1}(v))$
  \ENDIF
\ENDFOR
\end{algorithmic}
\caption{The $\frac{3}{2}$-competitive algorithm, \FPA, with advice.}
\label{algo:borrowing-advice}
\end{algorithm}

Since $\CEIL{\frac{\omega}{2}} \leq \CEIL{\frac{\OPT}{2}}$, we can
 give $\CEIL{\frac{\omega}{2}}$ as advice, resulting in
 Algorithm~\ref{algo:borrowing-advice}.
Note that the $f$-notation used in the pseudo-code
was defined in connection with
Algorithm~\ref{algorithm-optimal-advice-bipartite}.

 \begin{theorem}
\label{theorem:approx32}
 There is a $\frac{3}{2}$-competitive online algorithm for multi-coloring hexagonal graphs with advice complexity $\ENCODE{\CEIL{\frac{\OPT}{2}}}$.
\end{theorem}

\begin{proof}
 Given $\CEIL{\frac{\omega}{2}} \leq \CEIL{\frac{\OPT}{2}}$ as advice,
 \FPA can be used as an online algorithm
 (Algorithm~\ref{algo:borrowing-advice}).
\end{proof}

In~\cite{McDR00}, an algorithm with an improved approximation ratio of
$\frac{4}{3}$ was introduced.
We now describe this algorithm.
For completeness, we also give the arguments 
 that the algorithm is correct and is a $\frac43$-approximation algorithm:

The algorithm uses color classes in the same way as \FPA, except that the
 private color sets contain only $\FLOOR{\frac{\omega +1}{3}}$
 colors each.
We use the following notation.
For any node $v$, we let $n_v$ denote the number of requests to $v$.
Furthermore, $b_v$ denotes the maximum number of colors that
 $v$ can borrow, i.e., $b_v = \max \{ 0, \FLOOR{\frac{\omega +1}{3}}
-n'_v \}$,
 where $n'_v$ is the maximum number of requests to any of the neighboring nodes
in the color class that $v$ can borrow from.

The algorithm can be seen as working in up to three phases:

In the {\em first phase}, the algorithm colors $\min\{n_v,\FLOOR{\frac{\omega
 +1}{3}}\}$ requests to each node, $v$, using the node's
 private colors.
Let $G_1$ be the graph induced by the nodes that still have
 uncolored requests after Phase~1.

For any node, $v$, in $G_1$, $\FLOOR{\frac{\omega +1}{3}}$ requests to
 $v$ are colored with $v$'s private colors in Phase~1.
By the definition of $\omega$, this immediately implies that any pair
 of neighboring nodes have a total of at most $\omega - 2\FLOOR{\frac{\omega
 +1}{3}}$ uncolored requests already after Phase~1.

In the {\em second phase}, each node $v$ with more than
 $\FLOOR{\frac{\omega +1}{3}}$ requests borrows
 $\min\{n_v-\FLOOR{\frac{\omega +1}{3}}, b_v\}$ colors.
%
%Let $G_1$ be the graph induced by the nodes that still have uncolored
% requests after Phase 1.
%Similarly, 
Let $G_2$ be the graph induced by nodes that still have uncolored
 requests after Phase~2.

In~\cite{McDR00} it is proven that $G_2$ is bipartite and that any pair of
 neighbors in $G_2$ has a total of at most $\omega -
 2\FLOOR{\frac{\omega +1}{3}} \leq \FLOOR{\frac{\omega +1}{3}}+1$
 uncolored requests after Phase~2.
Thus, in the third phase, the remaining requests can be colored with \GEO
%(see Section~\ref{subsec:pathUpper})
(see the path section)
 using $\FLOOR{\frac{\omega +1}{3}}+1$ additional colors. 

To see that $G_2$ is bipartite, first note that $G_1$ (and hence
 $G_2$) cannot contain triangles.
Each node in such a triangle would have received at least
 $\FLOOR{\frac{\omega+1}{3}}+1$ requests, contradicting the definition
 of $\omega$.

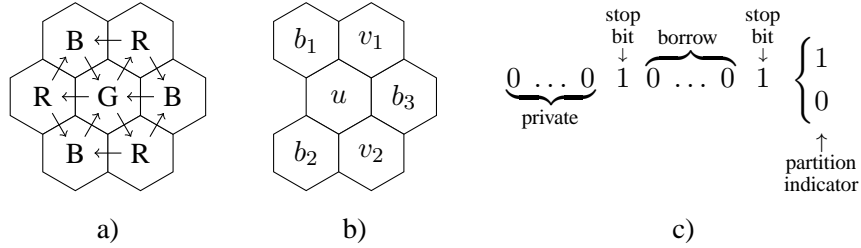
\begin{figure}
\begin{center}
%\begin{minipage}{.18\textwidth}
\begin{tabular}{c@{\hspace{2em}}c@{\hspace{2em}}c}
\begin{tikzpicture} [hexa/.style= {shape=regular polygon,regular polygon sides=6,minimum size=1cm, draw,inner sep=0,anchor=south,fill=white,rotate=30}]
 \node[hexa] (-1;1) at ({-0.500000*sin(60)},{0.750000}) {}; 
 \node[hexa] (-1;2) at ({0.000000*sin(60)},{1.500000}) {}; 
 \node[hexa] (0;0) at ({0.000000*sin(60)},{0.000000}) {}; 
 \node[hexa] (0;1) at ({0.500000*sin(60)},{0.750000}) {}; 
 \node[hexa] (0;2) at ({1.000000*sin(60)},{1.500000}) {}; 
 \node[hexa] (1;0) at ({1.000000*sin(60)},{0.000000}) {}; 
 \node[hexa] (1;1) at ({1.500000*sin(60)},{0.750000}) {}; 
 \node[circle,Black,minimum size=1cm](7) at (0;0) {B};
 \node[circle,Black,minimum size=1cm](3) at (0;1) {G};
 \node[circle,Black,minimum size=1cm](4) at (0;2) {R};
 \node[circle,Black,minimum size=1cm](1) at (-1;1) {R};
 \node[circle,Black,minimum size=1cm](2) at (-1;2) {B};
 \node[circle,Black,minimum size=1cm](6) at (1;0) {R};
 \node[circle,Black,minimum size=1cm](5) at (1;1) {B};
\path [draw, <-,shorten >=-3pt,shorten <=-3pt]
(1) edge (2)
(2) edge (3)
(3) edge (1)
(3) edge (4)
(4) edge (5)
(5) edge (3)
(4) edge (2)
(3) edge (6)
(6) edge (7)
(7) edge (3)
(1) edge (7)
(6) edge (5);
% \draw[arrows=<-,shorten >=-3pt,shorten <=-3pt](1)--(2)
% (2)--(3);
\end{tikzpicture}
%\end{minipage}
&
%\begin{minipage}{.21\textwidth}
\begin{tikzpicture} [hexa/.style= {shape=regular polygon,regular polygon sides=6,minimum size=1.0cm, draw,inner sep=0,anchor=south,fill=white,rotate=30}]
 \node[hexa] (1;2) at ({2.000000*sin(60)},{1.500000}) {}; 
 \node[hexa] (2;2) at ({3.000000*sin(60)},{1.500000}) {}; 
 \node[hexa] (1;1) at ({1.500000*sin(60)},{0.750000}) {}; 
 \node[hexa] (2;1) at ({2.500000*sin(60)},{0.750000}) {}; 
 \node[hexa] (0;3) at ({1.500000*sin(60)},{2.250000}) {}; 
 \node[hexa] (1;3) at ({2.500000*sin(60)},{2.250000}) {}; 
 \node [circle,Black,minimum size=1cm] at (1;2) {$u$};
 \node [circle,Black,minimum size=1cm] at (2;2) {$b_3$};
 \node [circle,Black,minimum size=1cm] at (0;3) {$b_1$};
 \node [circle,Black,minimum size=1cm] at (1;3) {$v_1$};
 \node [circle,Black,minimum size=1cm] at (1;1) {$b_2$};
 \node [circle,Black,minimum size=1cm] at (2;1) {$v_2$};
\end{tikzpicture}
%\end{minipage}
&
%\begin{minipage}{.36\textwidth}
\raisebox{8ex}{
$
  \underbrace{0\;\dots\;0}_\text{private} \;\; \SBIT \;\; \smash{\overbrace{0\;\dots\;0}^{\text{borrow~}} \;\; \SBIT} \;\;
  \begin{cases}
   1 &\\
   \PIND&
  \end{cases}
$
}
%\end{minipage}
\\
a) & b) & c)
\end{tabular}
\end{center}
\caption{Illustration of the $\frac43$-approximation algorithm.
a) The borrow pattern. Arrows show the direction of the flow of colors
   in Phase~2.
b) Part of a graph induced by nodes still having unprocessed requests
   after Phase~2.
c) The subsequence of advice bits connected to one node. The sequence of
   advice bits is a merge of such sequences.
}
\label{figure-induced-bipartite}
\end{figure}

Using the fact that $G_2$ does not contain triangles, we can now argue
 that $G_2$ is acyclic and hence bipartite. 
Assume to the contrary that $G_2$ does contain a cycle, $C$.
Assume without loss of generality that the \classred, \classgreen,
 \classblue coloring of the underlying hexagonal grid is as shown in
 Figure~\ref{figure-induced-bipartite} a) and let $u$ be a leftmost node of $C$.
Then, referring to Figure~\ref{figure-induced-bipartite} b), two of the
 nodes $v_1$, $v_2$, and $b_3$ must also be part of $C$. 
Note that $b_3$ cannot be part of $C$, since then there would be
 a triangle after Phase~$1$.
Thus, $u$, $v_1$, and $v_2$ are part of the cycle and hence receive at
 least $\FLOOR{\frac{\omega+1}{3}}+1$ requests each.

Since $u$ could not borrow enough colors from the nodes in the color class it is allowed
to borrow from, one of the $b$-nodes, say $b_j$, together with $u$
must have a total of at least $2\FLOOR{\frac{\omega+1}{3}}+1$ requests.
So, $b_j$ and $u$ must form a triangle together with either $v_1$ or $v_2$
so that the three nodes together have received a total of at
least $(2\FLOOR{\frac{\omega+1}{3}}+1) + (\FLOOR{\frac{\omega+1}{3}}+1)$
requests.
This quantity is strictly larger than $\omega$, contradicting the
definition of $\omega$.

This ends the argument that the algorithm is correct.

Since the total number of colors used is at most $3\FLOOR{\frac{\omega
 +1}{3}} + (\omega - 2\FLOOR{\frac{\omega +1}{3}}) \leq
 \frac{4\omega +1}{3}$, the algorithm is a $\frac43$-approximation
 algorithm.

\begin{algorithm}[thb]
\algsetup{indent=2em}
\begin{algorithmic}[1]
 \STATE {\bf Advice:} A sequence $B$ of
 bits classifying each request as to whether it should be colored using the node's own private colors, by borrowing, or in which partition it falls.
% \STATE \textbf{ADVICE $\advice$} read $\ENCODE{3^n}$ bits of Advice
% \STATE $\phasemin = \phasemax = 0$
 \STATE {\bf Function} $\class(v)$
 \STATE \qquad {\bf return} $v$'s color class: \classred, \classgreen, or \classblue
 \STATE {\bf Function} $\borrow(c)$
 \STATE \qquad {\bf return} the next class in the wrap-around sequence \classred, \classgreen, or \classblue
 \STATE {\bf Function} $\colors(c)$
 \STATE \qquad {\bf return} the set of private colors of class~$c$
 \STATE {\bf Function} NextBit($B$)
 \STATE \qquad {\bf return} the next advice bit
 \FOR{ each node $v$}
   \STATE $\text{Phase}(v) = \text{1}$
 \ENDFOR
 \FOR{ $i = 1$ \TO $n$}
 \STATE Assume that the $i$th request, $r$, is to node~$v$
 \IF{$\text{Phase}(v) = 1$}
   \IF{$\text{NextBit}(B)=0$}
%    \STATE /* $\maxcolor = 0$ \OR $\SIZE{f_{i-1}(v)} < \frac{3}{4}\maxcolor$ */
     \IF{$\colors(\class(v))\setminus f_{i-1}(v)=\emptyset$}
       \STATE add one color to each of the three sets of private colors
     \ENDIF
     \STATE give $r$ color $\min( \colors(\class(v)) \setminus f_{i-1}(v))$
   \ELSE
     \STATE $\text{Phase}(v) = 2$
     \STATE $\phasemin = 3 \, \SIZE{f_{i-1}(v)} +1$
     \STATE $\phasemax = 4 \, \SIZE{f_{i-1}(v)} +1$
   \ENDIF
 \ENDIF
 \IF{$\text{Phase}(v) = 2$}
   \IF{$\text{NextBit}(B)=0$}
     \STATE give $r$ color $\max( \colors(\borrow(\class(v)))
       \setminus f_{i-1}(v) )$
   \ELSE
     \STATE $\text{Phase}(v) = 3$
     \STATE $\text{upper}_v = \text{NextBit}(B)$  /* Store the partition of $v$ */
   \ENDIF
 \ENDIF
 \IF{$\text{Phase}(v) = 3$}
    \STATE /* Use \GEOA: */
    \IF{$\text{upper}_v=1$}
%       \STATE /* using the colors from $3\FLOOR{\frac{\omega+1}{3}}+1$ to $\maxcolor$ top-down*/
       \STATE give $r$ color $\max( \SET{\phasemin, \ldots,\phasemax} \setminus f_{i-1}(v))$
    \ELSE
%       \STATE /* using the colors from $3\FLOOR{\frac{\omega+1}{3}}+1$ to $\maxcolor$ bottom-up*/
%       \STATE /* $v$ is a lower node */
       \STATE give $r$ color $\min( \SET{\phasemin, \ldots,\phasemax} \setminus f_{i-1}(v))$
    \ENDIF
   \ENDIF
 \ENDFOR
\end{algorithmic}
\caption{Combining \FPA and \GEOA to a $\frac{4}{3}$-competitive algorithm.}
\label{algo:hex-43}
\end{algorithm}

We now show how an online algorithm, given the right advice, can
 behave as the off\-line $\frac43$-approximation algorithm.
Note that the three phases of the offline $\frac43$-approximation
 algorithm are characterized by the coloring strategy (using the
node's own private colors,
 borrowing private colors from neighbors, or coloring a bipartite graph).
However, when requests arrive online, the nodes may not go from one phase to
 the next simultaneously.

\begin{theorem}\label{theorem:43-hex}
There is a $\frac{4}{3}$-competitive online
   algorithm for multi-coloring hexagonal graphs with advice
   complexity at most $n + 2\SIZE{V}$.
\end{theorem}
\begin{proof}
We describe the algorithm and advice resulting in a coloring with at
 most $\frac43 \OPT$ colors
 (see Algorithm~\ref{algo:hex-43}, where we use the $f$-notation
 defined in connection with
 Algorithm~\ref{algorithm-optimal-advice-bipartite}).

Initially, each node is in Phase~$1$. On a request, the algorithm
reads an advice bit
and if it is zero, the next color from its private colors is used.
If, instead, a one is read, this is treated as a stop bit for Phase~$1$,
and this particular node enters Phase~$2$.

The algorithm starts with empty private color sets,
and adds one color to each set whenever necessary, i.e., whenever a
Phase~1 node that has already used all its private colors receives an
additional request (this includes the first request to the node).
As soon as a node leaves Phase~1,
the algorithm knows that this node received
$\FLOOR{\frac{\omega + 1}{3}}$ requests,
which is then the final size of each private color set.
Knowing the size of the private color sets, 
the algorithm can calculate the maximum color for the complete
coloring of the graph as $m = 4\FLOOR{\frac{\omega +1}{3}}+1$. 

In Phase~2, every zero indicates that the algorithm should borrow a color. When another stop bit
is received (which could be after no zeros at all if the borrowing
phase is empty), it moves to Phase~$3$.
In Phase~3, it reads one bit to decide which partition, upper or lower, of
the bipartite graph it is in, and does not need more information after that,
since it simply uses the colors $3\FLOOR{\frac{\omega +
    1}{3}}+1, \ldots, m$,
% assigned to Phase 3,
either top-down or bottom-up.

If we allow the algorithm one bit per request,
it needs at most two more bits per node, since the stop bits are the
only bits that do not immediately tell the algorithm which action to take.
Thus, $n + 2\SIZE{V}$ bits of advice suffice.
\end{proof}

%Algorithm~\ref{algo:hex-43}
This algorithm can be used in many different ways,
as long as the algorithm gets the information it needs.
One other simple encoding would be to give the algorithm the value
$\FLOOR{\frac{\omega + 1}{3}}$ from the beginning and 
only give bit-wise advice after a node has used all its private colors.
Since at least one color is private, this will save a total of at least $\SIZE{V}$ bits,
and result in at most
$\ENCODE{\FLOOR{\frac{\omega + 1}{3}}} + n + \SIZE{V}$ bits of advice.
This variant, and others, that are incomparable to each other,
depending on the values of $n$, $\omega$, and $\SIZE{V}$,
could all be used at the same time by first asking for a few bits
to decide how to proceed. Thus, one could formulate a less
readable but more accurate theorem basically taking the minimum
of all the expressions. We have chosen clarity over precision,
since the other expressions are mostly better in less interesting cases,
where $n$ is small compared to $\SIZE{V}$, for instance.

%Note that we could ensure that Algorithm~\ref{algo:hex-43} uses
% exactly the same number of colors as the $\frac43$-approximation
% algorithm.
%Each time a node leaves Phase 1, the algorithm could check the size of
% all cliques containing the node, updating $\omega$ if appropriate.
%Then, $\phasemax$ could be calculated as $\omega -
% 2\FLOOR{\frac{\omega + 1}{3}}$, as in the $\frac43$-approximation
% algorithm.

%\subsection*{Additional Remarks}
\subsection{Concluding Remarks}

When considering advice complexity of multi-coloring on a path,
we can achieve $1$-competitiveness with a small amount of advice.
A recoloring algorithm needs to be $1$-recoloring to achieve the same. 
The advice is basically the maximum number of requests to any
two neighboring nodes.
Thus, whether one has that global information once and for all,
or can obtain and adjust according to the local variant of this
information gives the same result.

For multi-coloring of hexagonal graphs, there is a similar connection
between recoloring distance and advice.
The $1$-recoloring online version of \FPA has an advice variant and again,
this advice represents information about the maximum number of requests
to neighboring nodes. 
With additional global information about the bipartite induced subgraph,
we can overcome the limitations of $1$-recoloring algorithms
and be as good as any known polynomial-time approximation algorithm.

\bibliography{refs-short}

\begin{thebibliography}{10}

\bibitem{AFG05}
S.~Albers, L.M. Favrholdt, and O.~Giel.
\newblock On paging with locality of reference.
\newblock {\em Journal of Computer and System Sciences}, 70(2):145--175, 2005.

\bibitem{BBFGHKSS14}
K.~Barhum, H.-J. B{\"o}ckenhauer, M.~Forisek, H.~Gebauer, J.~Hromkovi\v{c},
  S.~Krug, J.~Smula, and B.~Steffen.
\newblock On the power of advice and randomization for the disjoint path
  allocation problem.
\newblock In {\em SOFSEM}, volume 8327 of {\em LNCS}, pages 89--101. Springer,
  2014.

\bibitem{BBHK12}
M.~Paola Bianchi, H.-J. B{\"o}ckenhauer, J.~Hromkovic, and L.~Keller.
\newblock Online coloring of bipartite graphs with and without advice.
\newblock In {\em COCOON}, volume 7434 of {\em LNCS}, pages 519--530, 2012.

\bibitem{BKKK11}
H.-J. B{\"o}ckenhauer, D.~Komm, R.~Kr{\'a}lovi\v{c}, and R.~Kr{\'a}lovi\v{c}.
\newblock On the advice complexity of the $k$-server problem.
\newblock In {\em ICALP}, volume 6755 of {\em LNCS}, pages 207--218, 2011.

\bibitem{BKKKM09}
H.-J. B{\"o}ckenhauer, D.~Komm, R.~Kr{\'a}lovi\v{c}, R.~Kr{\'a}lovi\v{c}, and
  T.~M{\"o}mke.
\newblock On the advice complexity of online problems.
\newblock In {\em ISAAC}, volume 5878 of {\em LNCS}, pages 331--340, 2009.

\bibitem{BKKR12}
H.-J. B{\"o}ckenhauer, D.~Komm, R.~Kr{\'a}lovi\v{c}, and P.~Rossmanith.
\newblock On the advice complexity of the knapsack problem.
\newblock In {\em LATIN}, volume 6139 of {\em LNCS}, pages 61--72, 2012.

\bibitem{BIRS95}
A.~Borodin, S.~Irani, P.~Raghavan, and B.~Schieber.
\newblock Competitive paging with locality of reference.
\newblock {\em Journal of Computer and System Sciences}, 50(2):244--258, 1995.

\bibitem{BFLN03j}
J.~Boyar, L.M. Favrholdt, K.S. Larsen, and M.N. Nielsen.
\newblock {Extending the Accommodating Function}.
\newblock {\em Acta Informatica}, 40(1):3--35, 2003.

\bibitem{BGL12p}
J.~Boyar, S.~Gupta, and K.S. Larsen.
\newblock Access graphs results for {LRU} versus {FIFO} under relative worst
  order analysis.
\newblock In {\em SWAT}, volume 7357 of {\em LNCS}, pages 328--339. Springer,
  2012.

\bibitem{BKLL12}
J.~Boyar, S.~Kamali, K.S. Larsen, and A.~L{\'o}pez-Ortiz.
\newblock Online bin packing with advice.
\newblock In {\em STACS}, volume~25 of {\em LIPIcs}, pages 174--186. Schloss
  Dagstuhl -- Leibniz-Zentrum f{\"{u}}r Informatik GmbH, 2014.

\bibitem{BL99j}
J.~Boyar and K.S. Larsen.
\newblock {The Seat Reservation Problem}.
\newblock {\em Algorithmica}, 25(4):403--417, 1999.

\bibitem{BLN01j}
J.~Boyar, K.S. Larsen, and M.N. Nielsen.
\newblock {The Accommodating Function: a generalization of the competitive
  ratio}.
\newblock {\em SIAM Journal on Computing}, 31(1):233--258, 2001.

\bibitem{CCYZ10}
J.W.-T Chan, F.Y.L. Chin, D.~Ye, and Y.~Zhang.
\newblock Absolute and asymptotic bounds for online frequency allocation in
  cellular networks.
\newblock {\em Algorithmica}, 58(2):498--515, 2010.

\bibitem{CCYZZ06p}
J.W.-T. Chan, F.Y.L. Chin, D.~Ye, Y.~Zhang, and H.~Zhu.
\newblock Frequency allocation problems for linear cellular networks.
\newblock In {\em ISAAC}, volume 4288 of {\em LNCS}, pages 61--70. Springer,
  2006.

\bibitem{CFL13}
M.G. Christ, L.M. Favrholdt, and K.S. Larsen.
\newblock Online multi-coloring on the path revisited.
\newblock {\em Acta Informatica}, 50(5--6):343--357, 2013.

\bibitem{CJS13}
M.~Chrobak, L.~Jez, and J.~Sgall.
\newblock Better bounds for incremental frequency allocation in bipartite
  graphs.
\newblock {\em Theoretical Computer Science}, 514:75--83, 2013.

\bibitem{CS10}
M.~Chrobak and J.~Sgall.
\newblock Three results on frequency assignment in linear cellular networks.
\newblock {\em Theoretical Computer Science}, 411(1):131--137, 2010.

\bibitem{DKP09}
S.~Dobrev, R.~Kr{\'a}lovi\v{c}, and D.~Pardubsk{\'a}.
\newblock Measuring the problem-relevant information in input.
\newblock {\em RAIRO Theoretical Informatics and Applications}, 43(3):585--613,
  2009.

\bibitem{DHZ12}
R.~Dorrigiv, M.~He, and N.~Zeh.
\newblock On the advice complexity of buffer management.
\newblock In {\em ISAAC}, volume 7676 of {\em LNCS}, pages 136--145, 2012.

\bibitem{E75}
P.~Elias.
\newblock Universal codeword sets and representations of the integers.
\newblock {\em IEEE Transactions on Information Theory}, 21(2):194--203, 1975.

\bibitem{EFKR11}
Y.~Emek, P.~Fraigniaud, A.~Korman, and A.~Ros{\'e}n.
\newblock Online computation with advice.
\newblock {\em Theoretical Computer Science}, 412(24):2642--2656, 2011.

\bibitem{FKS12}
M.~Forisek, L.~Keller, and M.~Steinov{\'a}.
\newblock Advice complexity of online coloring for paths.
\newblock In {\em LATA}, volume 7183 of {\em LNCS}, pages 228--239, 2012.

\bibitem{HKK10}
J.~Hromkovi\v{c}, R.~Kr\'{a}lovi\v{c}, and R.~Kr\'{a}lovi\v{c}.
\newblock Information complexity of online problems.
\newblock In {\em MFCS}, volume 6281 of {\em LNCS}, pages 24--36, 2010.

\bibitem{JKM99}
J.~Janssen, K.~Kilakos, and O.~Marcotte.
\newblock Fixed preference channel assignment for cellular telephone systems.
\newblock {\em IEEE Transactions on Vehicular Technology}, 48(2):533--541,
  1999.

\bibitem{JKNS00j}
J.~Janssen, D.~Krizanc, L.~Narayanan, and S.M. Shende.
\newblock Distributed online frequency assignment in cellular networks.
\newblock {\em Journal of Algorithms}, 36(2):119--151, 2000.

\bibitem{KMRS88j}
A.R. Karlin, M.S. Manasse, L.~Rudolph, and D.D. Sleator.
\newblock Competitive snoopy caching.
\newblock {\em Algorithmica}, 3:79--119, 1988.

\bibitem{KK11j}
D.~Komm and R.~Kr{\'a}lovi\v{c}.
\newblock Advice complexity and barely random algorithms.
\newblock {\em RAIRO Theoretical Informatics and Applications}, 45(2):249--267,
  2011.

\bibitem{KKM12}
D.~Komm, R.~Kr{\'a}lovi\v{c}, and T.~M{\"o}mke.
\newblock On the advice complexity of the set cover problem.
\newblock In {\em CSR}, volume 7353 of {\em LNCS}, pages 241--252, 2012.

\bibitem{BBHKS13}
{M. P. Bianchi and H.-J. B\"{o}ckenhauer and J. Hromkovi\v{c} and S. Krug and
  B. Steffen}.
\newblock On the advice complexity of the online $l(2,1)$-coloring problem on
  paths and cycles.
\newblock In {\em COCOON}, volume 7936 of {\em LNCS}, pages 53--64. Springer,
  2013.

\bibitem{McDR00}
C.~McDiarmid and B.A. Reed.
\newblock Channel assignment and weighted coloring.
\newblock {\em Networks}, 36(2):114--117, 2000.

\bibitem{N02}
L.~Narayanan.
\newblock {\em Channel Assignment and Graph Multicoloring}, pages 71--94.
\newblock John Wiley \& Sons, Inc., 2002.

\bibitem{NS01}
L.~Narayanan and S.M. Shende.
\newblock Static frequency assignment in cellular networks.
\newblock {\em Algorithmica}, 29(3):396--409, 2001.

\bibitem{NS02}
L.~Narayanan and S.M. Shende.
\newblock Corrigendum: Static frequency assignment in cellular networks.
\newblock {\em Algorithmica}, 32(4):679, 2002.

\bibitem{SSU13}
S.~Seibert, A.~Sprock, and W.~Unger.
\newblock Advice complexity of the online coloring problem.
\newblock In {\em CIAC}, volume 7878 of {\em LNCS}, pages 345--357, 2013.

\bibitem{ST85j}
D.D. Sleator and R.E. Tarjan.
\newblock Amortized efficiency of list update and paging rules.
\newblock {\em Communications of the ACM}, 28(2):202--208, 1985.

\bibitem{SZ05j}
P.~Sparl and J.~Zerovnik.
\newblock 2-local 4/3-competitive algorithm for multicoloring hexagonal graphs.
\newblock {\em Journal of Algorithms}, 55(1):29--41, 2005.

\bibitem{WZ11p}
R.~Witkowski and J.~Zerovnik.
\newblock 1-local 33/24-competitive algorithm for multicoloring hexagonal
  graphs.
\newblock In {\em WAW}, volume 6732 of {\em LNCS}, pages 74--84, 2011.

\end{thebibliography}
\bibliographystyle{plain}

\end{document}